\documentclass[12pt]{article}

\usepackage{amsmath}
\usepackage{amsfonts}
\usepackage{amssymb}
\usepackage{graphicx}
\usepackage[english]{babel}
\usepackage{amsthm}
\usepackage{booktabs}
\usepackage[margin=1in]{geometry}
\usepackage{tabto}
\usepackage{latexsym}
\usepackage{comment}
\usepackage{enumitem}
\usepackage{float}
\usepackage{caption}
\usepackage{subcaption}
\usepackage{multicol}
\usepackage[utf8]{inputenc}
\usepackage[english]{babel}
\usepackage{setspace}
\usepackage[toc,page]{appendix}
\usepackage{mathtools}
\usepackage{cite}
\usepackage{hyperref}
\usepackage{authblk}
\usepackage[utf8]{inputenc}
\hypersetup{
    colorlinks=true,
    linkcolor=blue,
    citecolor=blue,
    urlcolor=cyan
}

\newtheorem{lemma}{Lemma}
\newtheorem{proposition}{Proposition}

\newtheorem{construction}{Construction}
\newtheorem{conjecture}{Conjecture}

\graphicspath{ {./images/} }

\title{
    \vbox{\LARGE\bfseries On the Weight Spectrum of \\ Reed-Muller Codes $RM(7,14)$\\[0.3em]}
}
\author{Milo Leuenberger\thanks{Corresponding author: miloleuenberger@bluewin.ch}, \  Manuel Albrizzio\thanks{Corresponding author: ManuelAlbrizzio@ferris.edu}\\
    \small Dept. of Mathematics, Ferris State University, Michigan, USA
}

\date{\small\textit{\today}}
\begin{document}
\maketitle
\begin{abstract}
In this paper, we attempt to find the weight spectrum of the Reed-Muller codes $RM(m-7,m)$. In the process, we tried to fully determine the weight spectrum of $RM(7,14)$.  We found almost all of the weights, except for a select number of them. The remaining eight missing weights are necessary to give a positive answer to an open question on the weight 
spectrum of $RM(m-c,m)$ for $ c= 7$. 
\end{abstract}

\subsection*{1 Introduction}

Reed-Muller codes belong to the family of error-correcting codes. Each codeword can be defined in terms of a Boolean function. The Reed-Muller codes $RM(r,m)$ have algebraic degree of at most $r$ and length $2^m$. The letter $m$ stands for the number of variables $x_1,...,x_m$ which take the values 0 or 1. The function $f(x)=f(x_1,...,x_m)\in \mathbb{F}_2[x_1,...,x_m]$ is then called a Boolean function. The codewords of the Reed-Muller code $RM(r,m)$ are the truth tables of the set of all ($m$-variable) Boolean functions which are polynomials of degree at most $r$~\cite{krupka1977introduction}. 

Knowing the weights of these functions is important to determine the distribution of codewords and how codes perform in error correction and data transmission. This information is useful for improving the efficiency of the decoding algorithms and for evaluating their complexities\cite{carlet2021boolean}.\\

The foundation for later work on determining the weight distribution of Reed–Muller codes was laid by Kasami, Tokura, and Azumi.
They derived explicit formulas for codeword weights less than 2.5$d$, where $d$ is the minimum weight~\cite{kasami1976weight}. 

In 2024, Carlet determined the weight spectrum of $RM(m-5,m)$ using the Maiorana-McFarland construction~\cite{carlet2024weight}. Shortly after, Lou and Wang determined the weight spectrum of $RM(m-6,m)$ using a concatenation technique~\cite{lou2025determining}. 
In this paper, we are attempting to determine the weight spectrum of $RM(m-7,m)$. We use the concatenation technique from Lou and Wang~\cite{lou2025determining} and extend it to four-term Boolean functions. 

The structure of the paper is build up in the following way. In Section 2, we introduce the necessary notation and mathematical method. In Section 3, we determine the weight spectrum of $RM(7,14)$ through a series of lemmas and propositions that we combine. We also note the missing weights. Section 4 concludes the paper with some remarks.  

\subsection*{2 Preliminaries}
Boolean functions are a map from $\mathbb{F}^m_2$ into the binary field $\mathbb{F}_2$. They can be represented in algebraic normal form (ANF). Let $f$ be a Boolean function with $m$ variables. Then, there exists a unique multivariate polynomial in $\mathbb{F}_2[x_1,...,x_m]$ such that

$$f(x_1, \ldots, x_m) = \sum_{L \subseteq \{1,2,\ldots,m\}} a_L \prod_{l \in L} x_l.$$

Let $B_m$ be the set of all Boolean functions with $m$ variables. $RM(r, m)$ represents the Reed-Muller code of $r$-th order and length $2^m$, which corresponds to the truth table of Boolean functions of degree at most $r$ in $m$ variables. The cardinality of a set A is denoted as $|A|$ in this paper. We use the same definition of cardinality and Hamming weight as Lou and Wang did in~\cite{lou2025determining}. If
 $$1_f = \{x\in \mathbb{F}^m_2 \ |\ f(x) = 1\} \ \ \ \text{and} \ \ \  0_f = \{x \in \mathbb{F}^m_2 \ | \ f(x) = 0\},$$  
then $|1_f|$ is called the \textit{Hamming weight of $f$}. It will be denoted by $wt(f)$.
The weights of the Reed–Muller code $RM(r,m)$ are just the Hamming weights of all its codewords, meaning all functions $f \in B_m$ with degree at most $r$. The weight spectrum is simply the set of all different weights that appear.

In this paper, we apply the concatenation technique used by Lou and Wang~\cite{lou2025determining}. The concatenation of Boolean functions is denoted by $\parallel$ and works the following way.
For two functions $f_1, f_2 \in B_m$, their concatenation results in a function in $B_{m+1}$:
\begin{align*}
(f_1 \parallel f_2)(x_1, \dots, x_m, x_{m+1}) &= (x_{m+1} + 1)f_1(x_1, \dots, x_m) + x_{m+1}f_2(x_1, \dots, x_m).
\end{align*}
Similarly, the concatenation of four functions $f_1, f_2, f_3, f_4 \in B_m$ results in a function in $B_{m+2}$,
\begin{align*}
f_1 \parallel f_2 \parallel f_3 \parallel f_4 ={} &(x_{m+1} + 1)(x_{m+2} + 1)f_1 + x_{m+1}(x_{m+2} + 1)f_2 \\
& {} + (x_{m+1} + 1)x_{m+2}f_3 + x_{m+1}x_{m+2}f_4.
\end{align*}

\subsection*{3 Determining the Weight Spectrum of the Reed-Muller Codes $RM(7,14)$}

Following a similar technique to~\cite{lou2025determining}, we first construct functions in $B_{14}$. 

\begin{lemma}
    Let $ g = g_0||g_1||g_2||(g_1+g_2+g_3) $, where $g_0,g_3 \in RM(5,12)$ and $g_1,g_2 \in RM(6,12)$. Then $g \in RM(7,14)$.

   \begin{proof}
    We compute
    \begin{align*}
    \quad\
    g
    &= (1 + x_{13})(1 + x_{14}) g_0
        + x_{13}(1 + x_{14}) g_1
        + (1 + x_{13}) x_{14} g_2
        + x_{13} x_{14} (g_1 + g_2 + g_3)\\
    &= (1 + x_{13})(1 + x_{14}) g_0
        + x_{13} g_1
        + x_{13} x_{14} g_1
        + x_{14} g_2
        + x_{13} x_{14} g_2
        + x_{13} x_{14} g_1
        + x_{13} x_{14} g_2
        + x_{13} x_{14} g_3\\
    &= (1 + x_{13})(1 + x_{14}) g_0
        + x_{13} g_1
        + x_{14} g_2
        + x_{13} x_{14} g_3 .
    \end{align*}
    Clearly, $g \in RM(7,14)$.
    \end{proof}
\end{lemma}

\begin{construction}
    Let $g_1 = x_1x_2x_3x_4x_5x_6 + x_{i_1}x_{i_2}x_{i_3}x_{i_4}x_{i_5}x_{i_6}$ and 
    $g_2 = x_1x_2x_3x_4x_5x_6 + x_{i_7}x_{i_8}x_{i_9}x_{i_{10}}x_{i_{11}}x_{i_{12}}$, 
    where $g_1,g_2 \in B_{12}$ and $1\leq i_1,...,i_{12}\leq 12$. We then construct $g=0||g_1||g_2|(g_1+g_2)$. 
\end{construction}
By Lemma 1, $g \in RM(7,14)$. In Construction 1, $g_1$, $g_2$, and $g_1+g_2$ all consist of two monomials. The following Lemma determines the weights of such functions. 

\begin{lemma}
    Let $f = x_{i_1}x_{i_2}x_{i_3}x_{i_4}x_{i_5}x_{i_6}
    + x_{i_7}x_{i_8}x_{i_9}x_{i_{10}}x_{i_{11}}x_{i_{12}} \in B_{12}.$ Let $I = \{i_1,i_2,i_3,i_4,i_5, i_6\}$ and $J = \{i_7, i_8, i_9, i_{10}, i_{11}, i_{12}\}.$ If $|I \cap J| = c$. Then, $$wt(f)=128-2^{c+1}.$$
    \begin{proof}
        Let $f_1 = x_{i_1}x_{i_2}x_{i_3}x_{i_4}x_{i_5}x_{i_6}$ and $f_2 = x_{i_7}x_{i_8}x_{i_9}x_{i_{10}}x_{i_{11}}x_{i_{12}}.$ Then,
        \begin{align*}
            \quad\
            wt(f)
            & = |1_{f_1}\cap0_{f_2}| + |0_{f_1}\cap1_{f_2}|\\
            & = 2^{|I\cap J|}(2^{6-|I\cap J|}-1) + 2^{|I\cap J|}(2^{6-|I\cap J|}-1)\\
            & = 2^c(2^{6-c}-1) + 2^c(2^{6-c}-1) \\
            & = 128 - 2^{c+1}.  
        \end{align*}
    \end{proof}
\end{lemma}
By Lemma 2, iterating through values for c, the set of all weights for functions of the form $f = x_{i_1}x_{i_2}x_{i_3}x_{i_4}x_{i_5}x_{i_6}
    + x_{i_7}x_{i_8}x_{i_9}x_{i_{10}}x_{i_{11}}x_{i_{12}} \in B_{12}$ is $\{0,64,96,112,120,124,126\}.$

\begin{lemma}
    The following weights, congruent with 2 modulo 4, can be obtained from Construction 1,   
    $$\{314, 342, 350\}.$$
     \begin{proof}
        Let $g = 0||g_1||g_2||(g_1+g_2)$, where $g_1 = x_1x_2x_3x_4x_5x_6+x_1x_2x_3x_4x_5x_7$ and $g_2=x_1x_2x_3x_4x_5x_6+x_6x_8x_9x_{10}x_{11}x_{12}$.
        By Lemma 1, $g \in RM(7,14).$ Also, $g_1 + g_2 = x_1x_2x_3x_4x_5x_7 + x_6x_8x_9x_{10}x_{11}x_{12}$. By Lemma 2, $wt(g_1)=128-2^6=64$, $wt(g_2) = 128-2^2 =124,$ and $wt(g_1+g_2)=128-2^1=126$. Thus, $wt(g) = 64+124+126 = 314.$
        \\
        \\
        Similarly, take $g_1 = x_1x_2x_3x_4x_5x_6+x_1x_2x_3x_4x_7x_8$ and $g_2=x_1x_2x_3x_4x_5x_6+x_5x_6x_9x_{10}x_{11}x_{12}$.
        By Lemma 2, $wt(g_1)=128-2^5=96$, $wt(g_2) = 128-2^3 =120,$ and $wt(g_1+g_2)=128-2^1=126$. Hence, $wt(g) = 96+120+126 = 342.$
        \\
        \\
        Finally, take $g_1 = x_1x_2x_3x_4x_5x_6+x_1x_2x_3x_7x_8x_9$ and $g_2=x_1x_2x_3x_4x_5x_6+x_4x_5x_6x_{10}x_{11}x_{12}$.
        By Lemma 2, $wt(g_1)=128-2^4=112$, $wt(g_2) = 128-2^4 =112,$ and $wt(g_1+g_2)=128-2^1=126$. Thus, $wt(g) = 112+112+126 = 350.$
     \end{proof}
\end{lemma}

Construction 1 can be extended to two three monomial functions.
\begin{construction} 
Let $g_1 = x_{1}x_{2}x_{3}x_{4}x_{5}x_{6}
    +x_{i_1}x_{i_2}x_{i_3}x_{i_4}x_{i_5}x_{i_6}
    + x_{i_7}x_{i_8}x_{i_9}x_{i_{10}}x_{i_{11}}x_{i_{12}}$ and $g_2 = x_{1}x_{2}x_{3}x_{4}x_{5}x_{6}
    +x_{i_1}x_{i_2}x_{i_3}x_{i_4}x_{i_5}x_{i_6}
    + x_{i_{13}}x_{i_{14}}x_{i_{15}}x_{i_{16}}x_{i_{17}}x_{i_{18}}$, where $g_1, g_2 \in B_{12}$ and $1\leq i_1,...,i_{18}\leq 12$. We then construct $g=0||g_1||g_2|(g_1+g_2)$.
\end{construction}
$g_1$ and $g_2$ in Construction 2 consist of three monomials, and $g_1+g_2$ consists of two monomials. By Lemma 1, $g \in RM(7,14)$. Let $J = \{i_7,i_8,i_9,i_{10},i_{11}, i_{12}\}$, 
and $K = \{i_{13}, i_{14}, i_{15}, i_{16}, i_{17}, i_{18}\}.$ If $|J\cap K| = c$, then by Lemma 2, $wt(g_1 + g_2) = 128-2^{c+1}$. The following Lemma is going to determine the weights of functions with three monomials.

\begin{lemma}
    Let $f = x_{1}x_{2}x_{3}x_{4}x_{5}x_{6}
    +x_{i_1}x_{i_2}x_{i_3}x_{i_4}x_{i_5}x_{i_6}
    + x_{i_7}x_{i_8}x_{i_9}x_{i_{10}}x_{i_{11}}x_{i_{12}}
    \in B_{12}.$ Let $I = \{1,2,3,4,5,6\},I = \{i_1,i_2,i_3,i_4,i_5, i_6\}$, 
    and $K = \{i_7, i_8, i_9, i_{10}, i_{11}, i_{12}\}.$ 
    If $|I \cap J| = c_1, |I\cap K| = c_2, |J\cap K| = c_3, and \, |I\cap J \cap K|= c_4$. Then, \\ 
    $$\quad wt(f)=192+2^{c_1+c_2+c_3-c_4-4}-2^{c_1+1}-2^{c_2+1}-2^{c_3+1}$$
    \begin{proof}
        Let
        $f_1 = x_{1}x_{2}x_{3}x_{4}x_{5}x_{i_6}, f_2 = x_{i_1}x_{i_2}x_{i_3}x_{i_4}x_{i_5}x_{i_6}$, and $f_3 = x_{i_7}x_{i_8}x_{i_9}x_{i_{10}}x_{i_{11}}x_{i_{12}}.$ 
        Then,
            $$ wt(f) = |1_{f_1}\cap1_{f_2}\cap 1_{f_3}| + |1_{f_1}\cap0_{f_2}\cap 0_{f_3}|+ |0_{f_1}\cap1_{f_2}\cap0_{f_3}| + |0_{f_1}\cap0_{f_2}\cap 1_{f_3}|.$$
        For the first term, by the principle of inclusion-exclusion $|I\cup J \cup K|=18-c_1-c_2-c_3+c_4$. Thus,
        \begin{align*}
            |1_{f_1}\cap1_{f_2}\cap 1_{f_3}|
            & = 2^{12-|I\cup J \cup K|}\\
            & = 2^{c_1+c_2+c_3-c_4-6}  
        \end{align*}
        For the second term, $|1_{f_1}\cap0_{f_2}|=|1_{f_1}\cap0_{f_2}\cap 0_{f_3}|+ |1_{f_1}\cap0_{f_2}\cap1_{f_3}|$. Therefore, 
        \begin{align*}
            \quad\
            |1_{f_1}\cap0_{f_2}\cap 0_{f_3}|
            & = |1_{f_1}\cap0_{f_2}|-|1_{f_1}\cap0_{f_2}\cap1_{f_3}|\\
            & = 64-2^{c_1}-|1_{f_1}\cap0_{f_2}\cap1_{f_3}|, \ \text{by Lemma 2}\\
            & = 64-2^{c_1} - (2^{6-c_1-c_3+c_4}-1)2^{c_1+c_2+c_3-c_4-6}\\
            & = 64-2^{c_1} - 2^{c_2}+2^{c_1+c_2+c_3-c_4-6}.
        \end{align*}
        Similarly,
            $$ |0_{f_1}\cap1_{f_2}\cap 0_{f_3}|= 64-2^{c_1} - 2^{c_3}+2^{c_1+c_2+c_3-c_4-6}$$
            $$ |0_{f_1}\cap0_{f_2}\cap 1_{f_3}|= 64-2^{c_2} - 2^{c_3}+2^{c_1+c_2+c_3-c_4-6}.$$ 
        Therefore, 
            \begin{align*}
                \quad\
                wt(f)
                & = 3(64)-2(2^{c_1}+2^{c_2}+2^{c_3})+4(2^{c_1+c_2+c_3-c_4-6})\\
                & = 192+2^{c_1+c_2+c_3-c_4-4}-2^{c_1+1}-2^{c_2+1}-2^{c_3+1}. 
            \end{align*}
    \end{proof}
\end{lemma}
By Lemma 4, iterating through all possible polynomials with 3 monomials, the set of all weights for functions of the form $x_{1}x_{2}x_{3}x_{4}x_{5}x_{6}
+x_{i_1}x_{i_2}x_{i_3}x_{i_4}x_{i_5}x_{i_6}
 + x_{i_7}x_{i_8}x_{i_9}x_{i_{10}}x_{i_{11}}x_{i_{12}}\in B_{12}$ is \\
$\{64,96,112,120,124,126,128,136,144,148,152,154,156,160,162,164,168,172,176\}.$\\
We now have a formula to compute the weight of functions with 3 monomials. Using Construction 2, we can achieve a weight of 254.

\begin{lemma}
    Let $g = 0||g_1||g_2||(g_1+g_2)$, where $g_1 = x_1x_2x_3x_4x_5x_6+x_1x_7x_8x_9x_{10}x_{11}+x_1x_7x_8x_9x_{10}x_{11}$ and $g_2=x_1x_2x_3x_4x_5x_6+x_1x_7x_8x_9x_{10}x_{11}+x_7x_8x_9x_{10}x_{11}x_{12}$. Then $g\in RM(7,14)$ and $wt(g)=254$.
     \begin{proof}
         By Lemma 1, $g \in RM(7,14).$ Also, $g_1 + g_2 = x_1x_7x_8x_9x_{10}x_{11} + x_7x_8x_9x_{10}x_{11}x_{12}$. By Lemma 2, $wt(g_1+g_2)=128-2^6=64$. By Lemma 4,  $wt(g_1)=192+2^3-2^2-2^2-2^7=64,$ and
         $wt(g_2) = 192+2^2-2^2-2^1-2^6=126$. Thus, $wt(g) = 64+126+64 = 254.$
     \end{proof}
\end{lemma}
Following the example in Lemma 5, we compute the weights of all possible constructions that are allowed within the constraints of Construction 2. We achieved this by implementing a computer algorithm, which iterates through all allowed combinations of functions and calculates the associated weights. 

\begin{lemma} 
The set of all new weights (2 mod 4) that can be created by Construction 2 is 
\[
\begin{aligned}
\{ &254, 318, 338, 346, 358, 366, 370, 374, 378, 386, 390, 394, 398, \\
   &402, 406, 410, 414, 418, 422, 426, 430, 434, 438, 442, 446, 450\}.
\end{aligned}
\]
\end{lemma}
More weights can be obtained from Construction 2 by flipping bits. This is discussed in the following propositions.

\begin{proposition}
The set of the
weights of functions in $RM(7,14)$ of the form $0||(g_1 +a_1)||(g_2)||(g_1 + g_2 +a_2)$, where $a_1, a_2 \in \mathbb{F}_2$, can generate the following set of weights congruent with 2 modulo 4. 
$$S_{2mod4, 2}=\{4238 + 4j \ | \ 0\leq j \leq 15, j \in \mathbb{Z}\} \cup \{4358,4362\} $$
\begin{proof}
    Let $g_1 = x_1x_2x_3x_4x_5x_6 + x_1x_2x_3x_7x_8x_{9} + f$ and $g_2 = x_1x_2x_3x_4x_5x_6 + x_1x_2x_3x_7x_8x_{9} + h$, where $g_1,g_2 \in B_{12}$ and $f$, $h$ are monomials of degree 6.

\begin{table}[H]
\centering
\label{tab:a1_constructions}
\begin{tabular}{cccccc}
\toprule
$wt(g)$ & $f$ & $h$ & $wt(g_1 + 1)$ & $wt(g_2)$ & $wt(g_1 + g_2)$ \\
\midrule
4242 & $x_1x_2x_3x_4x_{10}x_{11}$ & $x_7x_8x_9x_{10}x_{11}x_{12}$ & 3960 & 162 & 120 \\
4254 & $x_1x_2x_3x_4x_{10}x_{11}$ & $x_5x_6x_7x_8x_9x_{12}$ & 3960 & 168 & 126 \\
4258 & $x_1x_2x_3x_4x_5x_{10}$ & $x_4x_5x_6x_{10}x_{11}x_{12}$ & 3984 & 162 & 112 \\
4262 & $x_1x_2x_3x_4x_7x_{10}$ & $x_5x_6x_8x_9x_{11}x_{12}$ & 3968 & 168 & 126 \\
4266 & $x_1x_2x_3x_4x_5x_7$ & $x_4x_5x_6x_{10}x_{11}x_{12}$ & 3984 & 162 & 120 \\
4270 & $x_1x_2x_3x_4x_5x_7$ & $x_7x_8x_9x_{10}x_{11}x_{12}$ & 3984 & 162 & 124 \\
4274 & $x_1x_2x_3x_4x_5x_{10}$ & $x_6x_7x_8x_9x_{11}x_{12}$ & 3984 & 164 & 126 \\
4278 & $x_1x_2x_3x_4x_5x_7$ & $x_6x_8x_9x_{10}x_{11}x_{12}$ & 3984 & 168 & 126 \\
\bottomrule
\end{tabular}
\caption{Constructions for $S_{2\text{mod}4, 2}$ with $a_1=1, a_2=0$}
\end{table}

\begin{table}[H]
\centering
\label{tab:a3_constructions}
\begin{tabular}{cccccc}
\toprule
$wt(g)$ & $f$ & $h$ & $wt(g_1)$ & $wt(g_2)$ & $wt(g_1 + g_2 + 1)$ \\
\midrule
4238 & $x_1x_2x_4x_5x_6x_{10}$ & $x_3x_7x_8x_9x_{11}x_{12}$ & 120 & 148 & 3970 \\
4246 & $x_1x_2x_3x_4x_5x_7$ & $x_7x_8x_9x_{10}x_{11}x_{12}$ & 112 & 162 & 3972 \\
4250 & $x_1x_2x_3x_4x_5x_7$ & $x_4x_5x_6x_{10}x_{11}x_{12}$ & 112 & 162 & 3976 \\
4282 & $x_1x_2x_3x_4x_{10}x_{11}$ & $x_4x_5x_6x_{10}x_{11}x_{12}$ & 136 & 162 & 3984 \\
4286 & $x_1x_2x_4x_{10}x_{11}x_{12}$ & $x_3x_5x_6x_7x_8x_9$ & 156 & 160 & 3970 \\
4290 & $x_1x_2x_4x_5x_7x_{10}$ & $x_4x_5x_6x_{10}x_{11}x_{12}$ & 144 & 162 & 3984 \\
4294 & $x_1x_2x_3x_{10}x_{11}x_{12}$ & $x_4x_5x_6x_7x_8x_9$ & 148 & 176 & 3970 \\
4298 & $x_1x_2x_4x_7x_{10}x_{11}$ & $x_4x_5x_6x_{10}x_{11}x_{12}$ & 152 & 162 & 3984 \\
4358 & $x_4x_5x_6x_7x_{10}x_{11}$ & $x_4x_5x_6x_{10}x_{11}x_{12}$ & 164 & 162 & 4032 \\
4362 & $x_4x_5x_6x_{10}x_{11}x_{12}$ & $x_4x_5x_7x_{10}x_{11}x_{12}$ & 162 & 168 & 4032 \\
\bottomrule
\end{tabular}
\caption{Constructions for $S_{2\text{mod}4, 2}$ with $a_1=0, a_2=1$}
\end{table}

Choosing suitable variables for $f$ and $h$ and calculating $wt(g) = wt(g_1+1) + wt(g_2) + wt(g_1 + g_2)$ by Lemma 2 and Lemma 4, the construction can generate the weights in Table 1. 
Additionally, by computing $wt(g) = wt(g_1) + wt(g_2) + wt(g_1 + 1)$ with Lemma 2 and Lemma 4, the weights in Table 2 can be constructed. 
\end{proof}
\end{proposition}
Based on Construction 2, we can also generate weight congruent with 0 modulo 4.

\begin{proposition}
The 14-variable Boolean functions of the form $0||(g_1)||(g_2)||(g_1+g_2+1)$ can produce the weights $4372$ and $4376$. 

\begin{proof}
     Let $g_1 = x_1x_2x_3x_4x_5x_6 + x_1x_2x_7x_8x_9x_{10} + x_3x_4x_7x_8x_{11}x_{12}$ and $g_2 = x_1x_2x_3x_4x_5x_6 + x_1x_2x_7x_8x_9x_{10} + h$, where $g_1,g_2 \in B_{12}$ and $h$ is a monomial of degree 6.
     \\
     By Lemma 4, $wt(g_1) = 172$. \\
     
     Take $h = x_3x_4x_7x_8x_{10}x_{12} $, then $wt(g_2) = 168$ and by Lemma 2 $wt(g_1 + g_2) = 64$. Hence,
        $$wt(0||g_1||g_2||(g_1 + g_2 + 1)) = 172 + 168 + 4096- 64 = 4372. $$

    If $h = x_3x_4x_7x_9x_{11}x_{12} $, then $wt(g_2) = 172$ and $wt(g_1 + g_2) = 64$. Thus, 
    $$wt(0||g_1||g_2||(g_1 + g_2 + 1)) = 172 + 172 + 4096- 64 = 4376. $$
\end{proof}
\end{proposition}

\begin{construction}
    Let $g_1 = x_{1}x_{2}x_{3}x_{4}x_{5}x_{6}
    +x_{i_1}x_{i_2}x_{i_3}x_{i_4}x_{i_5}x_{i_6}
    + x_{i_7}x_{i_8}x_{i_9}x_{i_{10}}x_{i_{11}}x_{i_{12}}$ and $g_2 = x_{1}x_{2}x_{3}x_{4}x_{5}x_{6}
    +x_{i_{13}}x_{i_{14}}x_{i_{15}}x_{i_{16}}x_{i_{17}}x_{i_{18}}
    + x_{i_{19}}x_{i_{20}}x_{i_{21}}x_{i_{22}}x_{i_{23}}x_{i_{24}}$, where $g_1, g_2 \in B_{12}$ and $1\leq i_1,...,i_{24}\leq 12$. We then construct $g=0||g_1||g_2|(g_1+g_2)$.
\end{construction} 
Thus, $g_1$ and $g_2$ in Construction 3 consist of three monomials and $g_1+g_2$ consists of four monomials.
The weights of $g_1$ and $g_2$ are determined by Lemma 4. The following Lemma is going to determine the weights of functions with four monomials. 

\begin{lemma}
Let $f= x_1x_2\dots x_6+x_{i_1}x_{i_2}\dots x_{i_6}+x_{i_7}x_{i_8}\dots x_{i_{12}}+x_{i_{13}}x_{i_{14}}\dots x_{i_{18}}$ and 
$I=\{1,2,\dots,6\}$, $J= \{i_1,i_2,\dots,\i_6\}$, $K=\{i_7, i_8, \dots, i_{12} \}$, and $L=\{ i_{13}, i_{14}, \dots, i_{18} \}$.
If $|I\cap J|=c_1$, $|I\cap K|=c_2$, $|I\cap L|=c_3$, $|J\cap K|=c_4$, $|J\cap L|=c_5$, $|K\cap L|=c_6$, $|I\cap J \cap K|=c_7$, $|I\cap J \cap L|=c_8$, $|I\cap K \cap L|=c_9$, $|J\cap K \cap L|=c_{10}$, and $|I\cap J \cap K\cap L|=c_{11}$, then
\[
\begin{aligned}
wt(f) = 256 &- 2^{c_1+1} - 2^{c_2+1} - 2^{c_3+1} - 2^{c_4+1} - 2^{c_5+1} - 2^{c_6+1} \\
&+ 2^{c_1+c_2+c_4-c_7-4} + 2^{c_1+c_3+c_5-c_8-4} + 2^{c_2+c_3+c_6-c_9-4} + 2^{c_4+c_5+c_6-c_{10}-4} \\
&- 2^{c_1+c_2+c_3+c_4+c_5+c_6-c_7-c_8-c_9-c_{10}+c_{11}-9}.
\end{aligned}
\]

\begin{proof}
Let $f_I= x_1x_2\dots x_6$, $f_J= x_{i_1}x_{i_2}\dots x_{i_6}$, $f_K=x_{i_7}x_{i_8}\dots x_{i_{12}}$ and $f_L= x_{i_{13}}x_{i_{14}}\dots x_{i_{18}}$.
Then 
\begin{align*}
wt(f) = {} & |1_{f_I}\cap 0_{f_J}\cap 0_{f_K}\cap 0_{f_L}|+|0_{f_I}\cap 1_{f_J}\cap 0_{f_K}\cap 0_{f_L}| \\
                  & +|0_{f_I}\cap 0_{f_J}\cap 1_{f_K}\cap 0_{f_L}|+|0_{f_I}\cap 0_{f_J}\cap 0_{f_K}\cap 1_{f_L}| \\
                  & +|1_{f_I}\cap 1_{f_J}\cap 1_{f_K}\cap 0_{f_L}|+|1_{f_I}\cap 1_{f_J}\cap 0_{f_K}\cap 1_{f_L}| \\
                  & +|1_{f_I}\cap 0_{f_J}\cap 1_{f_K}\cap 1_{f_L}|+|0_{f_I}\cap 1_{f_J}\cap 1_{f_K}\cap 1_{f_L}|.
\end{align*}

By the principle of inclusion-exclusion for 4 sets, $$|I\cup J \cup K \cup L|= 24 - c_1 - c_2 - c_3 - c_4 - c_5 - c_6 + c_7 + c_8 + c_9 + c_{10} - c_{11}.$$

For the fourth term, 
\begin{align*}
    |1_{f_I}\cap 1_{f_J}\cap 1_{f_K}\cap 0_{f_L}|
    & = (2^{|L-I-J-K|}-1) 2^{12-|I \cup J \cup K \cup L|}\\
    & = (2^{6-c_3 - c_5 - c_6 + c_8 + c_9 + c_{10} - c_{11}}-1) 2^{c_1+c_2+c_3+c_4+c_5+c_6-c_7-c_8-c_9-c_{10}+c_{11}-12}\\
    & = 2^{c_1+c_2+c_4-c_7-6}-2^{c_1+c_2+c_3+c_4+c_5+c_6-c_7-c_8-c_9-c_{10}+c_{11}-12}
 \end{align*}

Similarly, 
 \begin{align*}
|1_{f_I}\cap 1_{f_J}\cap 0_{f_K}\cap 1_{f_L}| &= 2^{c_1+c_3+c_5-c_8-6} - 2^{c_1+c_2+c_3+c_4+c_5+c_6-c_7-c_8-c_9-c_{10}+c_{11}-12} \\
|1_{f_I}\cap 0_{f_J}\cap 1_{f_K}\cap 1_{f_L}| &= 2^{c_2+c_3+c_6-c_9-6} - 2^{c_1+c_2+c_3+c_4+c_5+c_6-c_7-c_8-c_9-c_{10}+c_{11}-12} \\
|0_{f_I}\cap 1_{f_J}\cap 1_{f_K}\cap 1_{f_L}| &= 2^{c_4+c_5+c_6-c_{10}-6} - 2^{c_1+c_2+c_3+c_4+c_5+c_6-c_7-c_8-c_9-c_{10}+c_{11}-12}
\end{align*}
 
For the first term,
    $$|1_{f_I}\cap 0_{f_J}\cap 0_{f_K}\cap 0_{f_L}|= |1_{f_I}\cap0_{f_J}\cap 0_{f_K}|-|1_{f_I}\cap 0_{f_J}\cap 0_{f_K}\cap 1_{f_L}|$$

Furthermore,
    $$|1_{f_I}\cap 0_{f_J}\cap 0_{f_K}\cap 1_{f_L}|= |1_{f_I}\cap0_{f_J}\cap 1_{f_L}|-|1_{f_I}\cap 0_{f_J}\cap 1_{f_K}\cap 1_{f_L}|$$

Therefore, 
    $$|1_{f_I}\cap 0_{f_J}\cap 0_{f_K}\cap 0_{f_L}|= |1_{f_I}\cap0_{f_J}\cap 0_{f_K}|-|1_{f_I}\cap0_{f_J}\cap 1_{f_L}|+|1_{f_I}\cap 0_{f_J}\cap 1_{f_K}\cap 1_{f_L}|$$
\[
\begin{aligned}
|1_{f_I} \cap 0_{f_J} \cap 0_{f_K}| &= 64 - 2^{c_1} - 2^{c_2} + 2^{c_1 + c_2 + c_4 - c_7 - 6}, \ \text{by Lemma 4.} \\
|1_{f_I} \cap 0_{f_J} \cap 1_{f_L}| &= 2^{c_3} - 2^{c_1 + c_3 + c_5 - c_8 - 6}
\end{aligned}
\]

Thus,
\begin{align*}
    |1_{f_I}\cap 0_{f_J}\cap 0_{f_K}\cap 0_{f_L}|
    & = 64 -2^{c_1} - 2^{c_2} -2^{c_3}\\ 
    & \ \ \ \  + 2^{c_1+c_2+c_4-c_7-6} + 2^{c_1+c_3+c_5-c_8 -6} + 2^{c_2+c_3+c_6-c_9-6}\\
    & \ \ \ \  - 2^{c_1+c_2+c_3+c_4+c_5+c_6-c_7-c_8-c_9-c_{10}+c_{11}-12}
\end{align*}

Similarly,
\begin{align*}
    |0_{f_I}\cap 1_{f_J}\cap 0_{f_K}\cap 0_{f_L}| 
    & = 64 -2^{c_1} - 2^{c_4} -2^{c_5}\\ 
    & \ \ \ \  + 2^{c_1+c_2+c_4-c_7-6} + 2^{c_4+c_5+c_6-c_{10} -6} + 2^{c_1+c_3+c_5-c_8-6}\\
    & \ \ \ \  - 2^{c_1+c_2+c_3+c_4+c_5+c_6-c_7-c_8-c_9-c_{10}+c_{11}-12}
\end{align*}
\begin{align*}
    |0_{f_I}\cap 0_{f_J}\cap 1_{f_K}\cap 0_{f_L}|  
    & = 64 -2^{c_2} - 2^{c_4} -2^{c_6}\\ 
    & \ \ \ \  + 2^{c_1+c_2+c_4-c_7-6} + 2^{c_2+c_3+c_6-c_9 -6} + 2^{c_4+c_5+c_6-c_{10}-6}\\
    & \ \ \ \  - 2^{c_1+c_2+c_3+c_4+c_5+c_6-c_7-c_8-c_9-c_{10}+c_{11}-12}
\end{align*}
\begin{align*}
    |0_{f_I}\cap 0_{f_J}\cap 0_{f_K}\cap 1_{f_L}|  
    & = 64 -2^{c_3} - 2^{c_5} -2^{c_6}\\ 
    & \ \ \ \  + 2^{c_4+c_5+c_6-c_{10}-6} + 2^{c_2+c_3+c_6-c_9 -6} + 2^{c_1+c_3+c_5-c_8-6}\\
    & \ \ \ \  - 2^{c_1+c_2+c_3+c_4+c_5+c_6-c_7-c_8-c_9-c_{10}+c_{11}-12}
\end{align*}

Putting it all together, we get
\[
\begin{aligned}
wt(f) = 256 &- 2^{c_1+1} - 2^{c_2+1} - 2^{c_3+1} - 2^{c_4+1} - 2^{c_5+1} - 2^{c_6+1} \\
&+ 2^{c_1+c_2+c_4-c_7-4} + 2^{c_1+c_3+c_5-c_8-4} + 2^{c_2+c_3+c_6-c_9-4} + 2^{c_4+c_5+c_6-c_{10}-4} \\
&- 2^{c_1+c_2+c_3+c_4+c_5+c_6-c_7-c_8-c_9-c_{10}+c_{11}-9}.
\end{aligned}
\]
\end{proof} 
\end{lemma}
  
By computing values with the formula in Lemma 7 for all allowed combinations of c's, the set of all weights for functions of the form $x_1x_2\dots x_6+x_{i_1}x_{i_2}\dots x_{i_6}+x_{i_7}x_{i_8}\dots x_{i_{12}}+x_{i_{13}}x_{i_{14}}\dots x_{i_{18}}\in B_{12}$ is,
\[
\begin{aligned}
\{ &0, 64, 96, 112, 120, 124, 126, 128, 136, 144, 148, 152, 154, 156, 158, 160, 162, 164, 166, 168, 172, \\
   &174, 176, 180, 182, 184, 186, 188, 190, 192, 194, 196, 198, 200, 202, 204, 208, 210, 212, 216, 220\}.
\end{aligned}
\]

\begin{lemma}
    Let $g = 0||g_1||g_2||(g_1+g_2)$, where $g_1 = x_1x_2x_3x_4x_5x_6+x_1x_2x_3x_4x_{7}x_{8}+x_1x_2x_3x_7x_{8}x_{9}$ and $g_2=x_1x_2x_3x_4x_5x_6+x_5x_6x_9x_{10}x_{11}x_{12}+x_5x_7x_9x_{10}x_{11}x_{12}$. Then $g\in RM(7,14)$ and $wt(g)=362$.
     \begin{proof}
         By Lemma 1, $g \in RM(7,14).$ Also, $g_1 + g_2 = x_1x_2x_3x_4x_{7}x_{8}+x_1x_2x_3x_7x_{8}x_{9}+x_5x_6x_9x_{10}x_{11}x_{12}+x_5x_7x_9x_{10}x_{11}x_{12}$. By Lemma 4,  $wt(g_1)=192+2^5-2^5-2^4-2^6=112,$ and $wt(g_2) = 192+2^3-2^3-2^2-2^6=124$. By Lemma 7, $wt(g_1+g_2)= 256-2^6-2^1-2^2-2^2-2^3-2^6+2^2+ 2^{2} + 2^{3} + 2^{2} + 2^{3}- 2^3=126$. Thus, $wt(g) = 112+124+126 = 362.$
     \end{proof}
\end{lemma}
Following the example in Lemma 8, we compute the weights of all possible constructions that are allowed within the constraints of Construction 3. To carry this out, we developed a suitable computer algorithm. This computation was significantly more expensive than the one for Construction 2. Therefore, we used matrices and permutations to cut down on computational complexity. 

\begin{lemma}
    The set of all new weights, congruent with 2 modulo 4, is
    \[
    \begin{aligned}
    \{ &362, 382, 454, 458 , 462, 466, 470, 474, 478, 482, 486, 490, 494,\\
    &498, 502, 506, 510, 514, 518, 522, 526, 530, 534, 538, 542, 546\}.
    \end{aligned}
    \]
\end{lemma}

\begin{proposition}
    Combining the weights of Lemma 3, 6, and 9, the following set $S_{2mod4,1}$ containing weights congruent with 2 modulo 4, can be generated.
    
    $$S_{2mod4,1} = \{254\} \cup \{314 + 4i\ | \ 0 \leq i \leq 58, i \in \mathbb{Z}\} \setminus M, $$
    where $M = \{322, 326, 330, 334, 354\}$ is the set of missing weights.  
\end{proposition}

\begin{lemma}
    ~\cite{carlet2024weight} The set of weights in RM(5,10) contains the weights $\{32, 48, 56, 60, 62, 64, 68\}$ and all even integers between 72 and 952. 
\end{lemma}

\begin{lemma}
    ~\cite{lou2025determining} Let A = $\{0,64,96,112,120,124,126,128,136,144,148\}$. Then the set S of weights in RM(6,12) is 
    $$S = A \cup\{152+2i\}\cup\{2^{12}-a\ |\ a\in A\},$$
    where i ranges over the set of consecutive integers from 0 to $2^{12}-152$.
\end{lemma}

\begin{proposition} 
Let $W_0 =\{0,128,192,224,240,248,252,254,256,272,288,296,304,308\}$ and $W_f$ be the weight spectrum of $RM(7,14)$.Then
$$W_0\cup\{312+2i\}\cup\{2^{14}- w\ | \ w \in W\}\setminus M \subseteq W_f,$$
where i ranges over the set of consecutive integers from 0 to $2^{13}-312$ and the set of missing weights $M = \{322, 326, 330, 334, 354, 8174, 8178, 8182, 8186, 8188, 8190, 8194,8196, 8198, 8202, \\ 8206, 8210, 16030,16050, 16054, 16058, 16062\}.$
    \begin{proof}
    Consider the functions of the form $g = g_0||g_1||g_2||(g_1+g_2+g_3)$, where $g_0, g_3\in RM(5,12)$ and $g_1,g_2\in RM(6,12)$. By Lemma 1, $g \in RM(7,14)$. Clearly, 
    $$wt(g) = wt(g_0)+wt(0||g_1||g_2||(g_1+g_2+g_3)).$$
    Let $h\in RM(5,10)$ and let $g_0$ be equal to $h$ within $RM (5,12)$. Then
    $$wt(g_0)=4wt(h).$$
    By Lemma 10, $wt(g_0)$ can achieve the weights in set $C$, 
    $$C = \{128, 192, 224, 240, 248, 256, 272\} \cup \{288+8i\},$$
    where $i$ ranges from 0 to 440. 
    \\
    \\
    \textbf{The weights 2mod4}
    \\
    Combining the sets from Proposition 1 and 3 with the set C as described above, we can achieve the following set of weights congruent with 2 modulo 4.
    $$S_{2mod4}=S_{2mod4,1}\cup\{s_1+c\ | \ s_1 \in S_{2mod4,1}, c \in C \} \cup S_{2mod4,2} \cup \{s_2+c\ | \ s_2 \in S_{2mod4,2}, c \in C\} $$
    $$S_{2mod4} = \{254\}\cup \{314 + 4i \ | \ 0\leq i \leq 1969, i \in \mathbb{Z} \}\setminus M'$$
    with the set of missing weights $M' =  \{322, 326, 330, 334, 354, 8174, 8178, 8182, 8186, 8190\}.$
    \\
    \\
    The set $S_{2mod4}$ can then be expanded by $wt(g + 1) = 2^{14}-wt(g)$ for $g \in RM(7, 14)$.
    $$S_{2mod4} = {254} \cup \{314 + 4i\} \cup \{2^{14}-254\} \setminus M,$$
    where $i$ ranges over the set of consecutive integers from 0 to $2^{12}-157$ and the set of missing weights $M=M' \cup \{2^{14}-M'\}.$
    \\
    \\
    \textbf{The weights 0mod4}
    \\
    If $f \in RM(6,12)$, then $g= 0||0||f||f \in RM(7,14)$ and $wt(g) = 2wt(f)$. Therefore, by Lemma 11, the following set $S_{0mod4,1}$ is also included:
    \\
    \\
    Let $B =\{0,128,192,224,240,248,252,256,272,288,296\}$, then 
    $$S_{0mod4,1} = B \cup \{304+4i\} \cup \{2^{13}-b\ |\ b\in B\},$$
    where i ranges over the set of consecutive integers from 0 to $2^{11}-152$. 
    \\
    \\
    Furthermore, the two weights from Proposition 2 can be combined with the set C in the same way as showed for the weights 2 mod 4. 
    $$ \{s + c \ | \ s \in \{4372, 4376\}, c \in C\}$$
    The critical weights from this set are included in $S_{0mod4,2}$.
    $$S_{0mod4,2} = \{7892 + 4i\ | \ 0\leq i\leq 73, i \in \mathbb{Z}\}$$
    Conclusively, we put $S_{0mod4,1}$ and $S_{0mod4,2}$ together.
        $$S_{0mod4}=S_{0mod4,1}\cup S_{0mod4,2}$$
        $$S_{0mod4} = B \cup \{304+4i\}$$
        where i ranges over the set of consecutive integers from 0 to $2^{11}-76.$ 
        Missing the weight 8188. 
    \\
    The set $S_{0mod4}$ can also be expanded by $wt(g + 1) = 2^{14}-wt(g)$ for $g \in RM(7, 14)$.
    $$S_{0mod4} = B \cup \{304+4i\} \cup \{2^{14}-b\ |\ b\in B\}$$
    where i ranges over the set of consecutive integers from 0 to $2^{12}-152$.
    Missing the weights 8188 and 8196. 
    \\
    \\
    \textbf{The final set}
    \\
    Ultimately, combining the sets $S_{0mod4}$ and $S_{2mod4}$ we get a subset of the weight spectrum of $RM(7,14).$
    $$S_{0mod4} \cup S_{2mod4}\subseteq W_f$$
    $$W_0\cup\{312+2i\}\cup\{2^{14}- w\ | \ w \in W_0\}\setminus M \subseteq W_f,$$
    where $i$ ranges over the set of consecutive integers from 0 to $2^{13}-312$\\
    and $W_0 = \{0,128,192,224,240,248,252,254,256,272,288,296,304,308\}.$ The list of missing weights $M$ is as mentioned above.
    \end{proof} 
\end{proposition} 
We are unsure if Proposition 4 is the total weight spectrum of $RM(7,14)$, given the conjecture from~\cite{carlet2024weight} about the weight spectrum of $RM(m-c,m)$. 

\begin{conjecture}
(Open question in ~\cite{carlet2024weight},~\cite{lou2025determining}) Let $c$ be any positive integer. For $m \geq 2c$, is the weight spectrum of $RM(m-c,m)$ of the form: 

$$ \{0\} \cup A \cup B \cup C \cup \overline{B} \cup \overline{A} \cup \{2^m\}$$

where, 
\begin{itemize}
    \item $A \subseteq [2^c, 2^{c+1}]$, is given by Kasami and Tokura~\cite{kasami1970weight},
    \item $B \subseteq [2^{c+1}, 2^{c+1} + 2^{c-1}]$, is given by Kasami, Tokura, and Azumi in~\cite{kasami1976weight},
    \item $C \subseteq [2^{c+1} + 2^{c-1}, 2^m - 2^{c+1} - 2^{c-1}]$, consists of all consecutive even integers,
    \item $\overline{A}$ stands for the complement to $2^m$ of $A$, and $\overline{B}$ stands for the complement to $2^m$ of $B$.
\end{itemize}
\end{conjecture}

\subsection*{4 Conclusion}
In this paper, we use the concatenation technique of Lou and Wang in an attempt to determine the weight spectrum of the Reed-Muller code $RM(7,14)$. Achieving this spectrum, is a crucial step in finding the weight spectrum of the Reed-Muller codes $RM(m-7,m)$. We adapted the functions in the concatenation technique by one degree and two variables. Furthermore, we extended the construction to the concatenation of 4-term Boolean functions. In the process, we derived a formula to determine the weight of a 4-term Boolean function in $RM(6,12)$. 

To satisfy Conjecture 1 for $m = 14$, we are missing a small number of weights in the spectrum of $RM(7,14)$. We were not able to determine the missing weights by using the concatenation technique used in this paper. Therefore, a different construction needs to be applied. Given the missing weights we will not generalize to $RM(m-7,m)$, because of the implied missing weights in larger Reed-Muller code spaces.
\\
\\
\textbf{Remark.}
Because of the way the weights are constructed, the list of missing weights $M$ could be achieved by finding the following 8 weights: 
$$\{322,326,330,334,354, 4378, 4380, 4382\}.$$

\renewcommand{\refname}{\subsection*{References}\vspace{-1.5ex}}

\end{document}